\documentclass[11pt]{article}

\usepackage{fullpage}
\usepackage{amsmath,amsfonts,amsthm,mathrsfs,mathpazo,xspace,hyperref,graphicx}
\usepackage{endnotes}
\usepackage{color}
\usepackage{bm}
\usepackage{times}
\usepackage{amssymb,latexsym}
\usepackage{verbatim}

\newtheorem{theorem}{Theorem}

\newtheorem{lemma}[theorem]{Lemma}

\newtheorem{corollary}[theorem]{Corollary}

\newcommand{\beq}{\begin{eqnarray}}
\newcommand{\eeq}{\end{eqnarray}}

\newcommand{\ket}[1]{|#1\rangle}
\newcommand{\bra}[1]{\langle#1|}
\newcommand{\Tr}{\mbox{\rm Tr}}
\newcommand{\Id}{\ensuremath{\mathop{\rm Id}\nolimits}}
\newcommand{\Es}[1]{\textsc{E}_{#1}}

\newcommand{\C}{\ensuremath{\mathbb{C}}}

\newcommand{\R}{\ensuremath{\mathbb{R}}}

\newcommand{\mH}{\mathcal{H}}

\newcommand{\reg}[1]{{\textsf{#1}}}
\def\ancilla {{\mathrm{extra}}}

\newcommand{\setft}[1]{\mathrm{#1}}

\newcommand{\Lin}{\setft{L}}

\DeclareMathOperator{\poly}{poly}

\newcommand{\eps}{\varepsilon}

\newcommand{\CHSH}{\ensuremath{\textsc{CHSH}}}

\newif\ifnotes\notesfalse

\begin{document}

\title{Entanglement of approximate quantum strategies in XOR games}
\author{Dimiter Ostrev\thanks{Department of Mathematics, Massachusetts Institute of Technology, USA.} \and Thomas Vidick\thanks{California Institute of Technology, Pasadena, USA. Research supported by NSF CAREER Grant CCF-1553477 and the IQIM, an NSF Physics Frontiers Center (NFS Grant PHY-1125565) with support of the
Gordon and Betty Moore Foundation (GBMF-12500028).}}
\date{}
\maketitle

\begin{abstract}
We show that for any $\eps>0$ there is an XOR game $G=G(\eps)$ with $\Theta(\eps^{-1/5})$ inputs for one player and $\Theta(\eps^{-2/5})$ inputs for the other player such that $\Omega(\eps^{-1/5})$ ebits are required for any strategy achieving bias that is at least a multiplicative factor $(1-\eps)$ from optimal. This gives an exponential improvement in both the number of inputs or outputs and the noise tolerance of any previously-known self-test for highly entangled states. 
Up to the exponent $-1/5$ the scaling of our bound with $\eps$ is tight: for any XOR game there is an $\eps$-optimal strategy using $\lceil \eps^{-1} \rceil$ ebits, irrespective of the number of questions in the game. 
\end{abstract}

\section{Introduction}

Perhaps the most striking demonstration of the radical departure of quantum systems from classical behavior is given by the Bell test. Recent experiments~\cite{hensen2015loophole,giustina2015significant,shalm2015strong} establish ``all-loopholes-closed'' validations of the simplest such test, the CHSH inequality~\cite{Clauser:69a}. Although they do not reach the maximum quantum bound of $2\sqrt{2}$, the observed violation and statistical confidence are high enough to provide a solid proof of quantumness of the underlying physical system. 

Research in quantum cryptography and self-testing in recent years has established that a  large violation of the CHSH inequality goes much further than a generic certificate of non-classical behavior: it can serve as a guarantee that the underlying quantum system is locally isometric to one that is in a Bell pair $\ket{\phi^+}=\frac{1}{\sqrt{2}}(\ket{00}+\ket{11})$. This can be interpreted as a form of ``self-test'' for the Bell pair, by which its presence is certified solely via observable correlations, irrespective of the measurements being made. 

Can more complex entangled states similarly be verified by the violation of a suitable Bell inequality? Due to its importance for experiments as well as quantum cryptography, the question has been well-studied. The most relevant state of the art for us is the following: for any dimension $d$ there exists a Bell inequality whose maximum violation by a quantum system can only be achieved if the system is locally isometric to a $d$-dimensional maximally entangled state~\cite{yang2013robust}. With the exception of the results from~\cite{Slofstra11xor} (to which we return in more detail below), however, all known self-tests for $d$-dimensional entangled states require either a number of inputs~\cite{yang2013robust,coalangelo} or outputs~\cite{mckague2016high} that scales at least linearly with $d$, i.e. the test has size exponential in the number of ebits tested.  

The situation is even less satisfying as soon as one attempts to certify an even slightly noisy system, where by noisy system we mean one that will only lead to a violation that approaches the quantum optimum up to a multiplicative factor $(1-\eps)$ for some $\eps>0$. The performance of known tests scales poorly with the ``robustness parameter'' $\eps$, which in virtually all cases is required to be inverse exponential in the number of ebits tested before any consequence can be drawn.\footnote{We survey the relevant results in more detail in Section~\ref{sec:related} below.} Is this dependence necessary? 

We study the question in the context of the simplest kind of Bell inequalities, two party binary output correlation inequalities. These are bipartite Bell inequalities where each site can be measured using any number of two outcome local observables, but only expectation values of the correlators of the outcomes obtained at each site are taken into account. Such inequalities can be equivalently formulated using the language of two-player XOR games, that we adopt from now on. An XOR game is a two-player one-round game $G$ in which the players' answers are restricted to be a single bit each, and the verifier's acceptance criterion only depends on the parity of these bits. Any binary output correlation inequality can be mapped into an XOR game and vice-versa. The bias $\beta^*$ of the XOR game, defined as twice the maximum deviation from $1/2$ of the players' success probability, is the quantity that plays the role of the quantum bound for the Bell inequality.

\subsection{Results}

Our main result is that XOR games can provide very efficient tests for high-dimensional entanglement, while at the same time being noise-robust --- to some extent. In the positive direction we show that for any $\eps>0$, there exists an XOR game with $\Theta(\eps^{-1/5})$ inputs for Alice and $\Theta(\eps^{-2/5})$ inputs for Bob such that any strategy that comes within a multiplicative $(1-\eps)$ of the optimal quantum bias $\beta^* = \sqrt{2}/2$ requires the use of a state that is close to a tensor product of $\Omega(\eps^{-1/5})$ EPR pairs. Thus both the number of settings and the certified number of ebits are inverse polynomial in $\eps$. (The number of outcomes, of course, is only two.) In the negative direction we show that, up to the exponents $-1/5$, no XOR game can lead to a better scaling: for any XOR game and any $\eps>0$ there exists a strategy coming within a multiplicative factor $(1-\eps)$ of the optimal bias that uses  $O(\eps^{-1})$ EPR pairs (irrespective of the number of inputs in the game). 

For our positive result we consider a family of XOR games introduced by Slofstra~\cite{Slofstra11xor}. For an integer $n\geq 2$, the game\footnote{This game should not be confused with the $\CHSH_q$ game introduced in~\cite{bavarian2015information}.} $\CHSH(n)$ has $n$ possible questions for Alice, indexed by integers $i\in\{1,\ldots,n\}$, and $n(n-1)$ possible questions for Bob, indexed by pairs  $(i,j)\in\{1,\ldots,n\}^2$ such that $i\neq j$. The game can be described as follows: the referee selects a pair $(i,j)\in \{1,\ldots,n\}^2$ such that $i\neq j$ uniformly at random. He sends either $i$ or $j$ to Alice (with probability $1/2$ each), and $(i,j)$ to Bob. The players have to provide answers $a,b\in\{0,1\}$ such that $a\oplus b = 1$ if $i>j$ and Alice received $i$, and $a\oplus b = 0$ in the remaining three cases. 

Note that $\CHSH(2)$ is the usual CHSH game, for which the optimal bias is $\beta^*(\CHSH)=\sqrt{2}/2$. Slofstra showed that $\beta^*(\CHSH(n)) = \sqrt{2}/2$ for all $n\geq 2$, and that strategies achieving the optimum bias in $\CHSH(n)$ require a Hilbert space of dimension $2^{\lfloor n/2 \rfloor}$. Our theorem implies a smooth degradation of this bound for $\eps>0$. 

\begin{theorem}\label{thm:main}
Let $\eps>0$, let $n=\Theta(\eps^{-1/5})$ be an integer and $(A_i,B_{ij},\ket{\psi})$ a strategy in $\CHSH(n)$ achieving bias at least $(1-\eps)\beta^*(\CHSH(n))$. Then $\ket{\psi}$ has entanglement entropy $\Omega(\eps^{-1/5})$.  
\end{theorem}

Switching the parameters around, Theorem~\ref{thm:main} implies in particular that for any integer $n\geq 2$ and $\eps = O(n^{-5})$, any $\eps$-optimal strategy in $\CHSH(n)$ requires entanglement of dimension $2^{\Omega(n)}$.
The proof of Theorem~\ref{thm:main} in fact yields a stronger "rigidity'' result for the game $\CHSH(n)$, showing that for any strategy achieving bias at least $(1-\eps)$ times the optimum in $\CHSH(n)$ and any $r\leq \lfloor n/3 \rfloor$ there are local isometries that map the strategy to one that is within distance $O(r^{5/2}\sqrt{\eps})$ of a tensor product of $r$ ideal strategies for the game $\CHSH(2)$.\footnote{We refer to Section~\ref{sec:lb} for details.} 

Our negative result complements the lower bound from Theorem~\ref{thm:main}. We prove the following: 

\begin{theorem}\label{thm:ub}
Let $\eps>0$ and let $G$ be an XOR game. Then there exists an $\eps$-optimal strategy for $G$ using a maximally entangled state in $2^{\lceil \eps^{-1}\rceil}$ dimensions. 
\end{theorem}

The same result, with a slightly weaker upper bound $d = 2^{O(\eps^{-2})}$, is attributed to Regev in~\cite{CHTW04}. We nevertheless include a complete proof in Section~\ref{sec:ub}, as to the best of our knowledge the result had not previously appeared in print. 

\paragraph{Applications.} 
Our result can be interpreted as a robust, efficient self-test for the tensor product of $n$ EPR pairs: given any integer $n$, setting $\eps = O(n^{-5})$ any strategy in $\CHSH(3n)$ that achieves a bias at least $(1-\eps)$ times the optimal must be using a state that is close to an $n$-qubit maximally entangled state. The game $\CHSH(3n)$ only has $O(n^2)$ inputs per player, and it thus provides a very efficient test, with the number of inputs scaling only quadratically with the number of ebits tested. 

The work of Reichardt et al.~\cite{ReichardtUV13nature} demonstrates that self-testing results for the tensor product of many EPR pairs can form the basis for much more complex tasks, such as the classical delegation of an arbitrary quantum circuit to two isolated provers. It would be interesting to investigate whether the analysis of the $\CHSH(n)$ game that we give here could be leveraged to improve the efficiency of their protocol. Our self-testing result gives access to $n$ mutually anti-commuting pairs of observables on Alice's system, which can be combined to create arbitrary Pauli operators. Paulis of high weight will require taking the product of many observables, yielding a corresponding loss in error. However, one can easily imagine modifying the $\CHSH(n)$ game by introducing inputs associated with specific Pauli operators one is interested in. 

In~\cite{kaniewski2014entropic} the $\CHSH(n)$ game is used to test effective anti-commutators, from which a form of device-independent uncertainty relation can be derived. The stronger guarantees that come out of our analysis may have further applications to device-independent cryptography. 

\paragraph{Proof idea of Theorem \ref{thm:main}.}
We briefly discuss the proof of Theorem~\ref{thm:main}, referring to Section~\ref{sec:lb} for more details. Let $A_i$ (resp. $B_{ij}$) be Alice's (resp. Bob's) observables, and $\ket{\psi}$ the entangled state, in an $\eps$-optimal strategy for $\CHSH(n)$. Our proof proceeds in three steps. 

First we observe that $\CHSH(n)$ contains ${n \choose 2}$ copies of the $\CHSH$ game embedded inside it, one for each pair $\{i,j\} \subseteq \{1,\ldots,n\}$. By applying well-known rigidity results for the $\CHSH$ game we obtain approximate anti-commutation relations between each pair of Alice's observables. 

In the second step we show that any such $n$ pairwise approximately anti-commuting observables can be used to construct $m=\lfloor n/3 \rfloor$ pairs $(X_k,Z_k)$ of anti-commuting observables such that any two observables belonging to distinct pairs approximately commute. 
 
Finally, in the third and last step we show that the observables constructed in the second step can be interpreted as $m$ approximate overlapping qubits, where a qubit is defined as a pair of anti-commuting observables and two qubits are said to partially overlap if the associated observables approximately commute. We apply a theorem due to~\cite{reichardt16testing}, which shows that overlapping qubits are not far from exact qubits. The lower bound on entanglement entropy follows from an application of strong subadditivity and Fannes' inequality. 

\paragraph{Proof idea of Theorem \ref{thm:ub}.} We also briefly discuss the proof of Theorem~\ref{thm:ub}, referring to Section~\ref{sec:ub} for more details. 

As we mentioned before, the result of Theorem \ref{thm:ub} is attributed to Regev in~\cite{CHTW04}, with the slightly weaker upper bound $d = 2^{O(\eps^{-2})}$. The improvement from $\eps^{-2}$ to $\eps^{-1}$ requires a slightly more careful analysis of the performance of the randomly projected vectors in the semidefinite program associated to the XOR game. Although its implication for XOR games has not previously been spelled out, the improved bound is not new, and can be obtained in a number of different ways. For instance it follows from the analysis of Krivine rounding schemes in~\cite[Theorem 1.1]{naor2014krivine}, and was also obtained using Riesz's rounding technique in~\cite[Theorem 4]{Montanari16}. We provide a different analysis based on a rounding technique which was used in~\cite{NRV13} to analyze the non-commutative Grothendieck inequality and originates in Hirschman's proof~\cite{hirschman1952convexity} of the Hadamard three-line theorem in complex analysis.

\subsection{Related works}
\label{sec:related}

The general study of optimal strategies in XOR games was initiated by Tsirelson, who shows \cite{tsirel1987quantum} that for any XOR game with $n$ and $m$ inputs per party there is an optimal strategy that uses a maximally entangled state of dimension at most $2^{\lfloor r/2 \rfloor}$, where $r$ is the largest integer that satisfies $\binom{r+1}{2} \leq n+m$ and $r \leq \min(m,n)$. To establish this Tsirelson first proves that  to each player's input in the game can be associated a real $r$-dimensional unit vector, $x_i$ for Alice and $y_j$ for Bob, such that the correlations $x_i\cdot y_j$ achieve the optimal quantum bias in the game. Tsirelson then uses a clever construction, based on a representation of the Clifford algebra, to show that these vectors can be mapped to observables and a maximally entangled state in dimension $2^{\lfloor r/2 \rfloor}$ that achieve precisely the same correlations. Slofstra~\cite{Slofstra11xor} shows that Tsirelson's bound is tight for a slight variant of the $\CHSH(n)$ game. 

These results characterize the dimension of exactly optimal strategies in any XOR game. To the best of our knowledge, even if one considers arbitrary two-player games the $\CHSH(2n)$ game remains the most efficient (in terms of total number of inputs and outputs per party) test for $n$-qubit maximally entangled states. In particular, although there is strong indication that certain Bell inequalities, such as the $I_{3322}$ inequality, have a quantum bound that may only be achieved in the limit of infinite dimensions~\cite{PalV10I3322}, no such result has been rigorously proven.\footnote{There are examples of two-player one round games which provably require infinite-dimensional entanglement in order to be played optimally~\cite{LTW08,RegevV12a}, but these require the exchange of quantum messages between the referee and the players.} Recently Slofstra~\cite{slofstra2016tsirelson} showed the existence of a game for which a value $1$ can be attained using infinite-dimensional commuting-operator strategies, but it is not known if there exists a tensor product strategy achieving this value; in particular there is no ``optimal entangled state'' for this game. 

Lower bounds on entanglement become much weaker as soon as one considers strategies that only achieve a factor $(1-\eps)$ of the optimum. First we consider the case of XOR games. To the best of our knowledge, prior results focused on the dimension of the Hilbert space required for the strategy, which does not necessarily imply high entanglement entropy.\footnote{For any $\delta>0$, for any positive integer $n$, there exist states with Schmidt rank $n$ and entanglement entropy less than $\delta$.} The best prior lower bound on the dimension of the Hilbert space scales as $1/\eps$; precisely $\lceil 1/(2\eps)\rceil$~\cite{BrietBT11grot}. The bound proven in~\cite{BrietBT11grot} in fact applies to the dimension of the vectors that constitute an approximately optimal solution to the semidefinite program associated to an XOR game (see Section~\ref{sec:xor} for a definition). Another interesting work is~\cite{Slofstra11xor}, where approximate representations of $C^*$-algebras are used to establish lower bounds on the Hilbert space dimension needed for $\eps$-optimal strategies. The lower bound on dimension shown there scales as $\eps^{-1/12}$. In addition, \cite{Slofstra11xor} proves the lower bound $n-8\sqrt{2}n(n-1)\eps$ for the dimension of the vectors that constitute an approximately optimal solution to the semidefinite program associated to the $\CHSH(n)$ game.

Recent results derive better bounds for approximately optimal strategies by considering more general two-player games than XOR games. A natural approach to testing an $n$-qubit maximally entangled state consists in considering games based on the parallel repetition of $n$ copies of (a slight variant of) the $\CHSH$ game~\cite{mckague2016self}; however this parallel repetition requires a number of inputs and outputs that is exponential in $n$; furthermore no good bounds are known on the noise tolerance of the resulting tests. (See very recent work~\cite{coalangelo} giving robustness bounds for parallel self-testing that scale as $\eps = \poly^{-1}(n)$; however the number of inputs needed still scales exponentially with $n$.) A more direct approach to testing $d$-dimensional maximally entangled states is given in~\cite{yang2013robust}, but here again the number of measurement settings scales linearly with the dimension $d$ and no explicit bound on the noise tolerance is given. Recently one of us~\cite{reichardt16testing} showed a lower bound of $2^n$ on the dimension of $O(n^{-3/2})$-optimal strategies for a game with $O(n)$ questions per player that is similar to the $\CHSH(n)$ game but is not an XOR game. In this paper we re-use one of the main technical contributions of~\cite{reichardt16testing}, Theorem~\ref{t:EPRstabilizersfrommath}.

\section{Preliminaries}

\subsection{Notation} 

For a set $S$ we write $\Es{i\in S}$ for $|S|^{-1}\sum_{i\in S}$. All Hilbert spaces in this paper are finite-dimensional; we use a calligraphic letter $\mH,\mH_\reg{A},\mH_\reg{B}$ to denote a finite-dimensional Hilbert space. Given $A\in\Lin(\mathcal{H})$ the absolute value $|A|$ is defined as the unique positive square root of $A^\dagger A$. For $A\in\Lin(\mH)$ we write $A^{-1}$ or (when there is no ambiguity) $\frac{1}{A}$ for the Moore-Penrose pseudo-inverse of $A$. 

An observable is a Hermitian operator $A\in\Lin(\mH)$ that squares to identity. We will call an observable balanced if its 1-eigenspace and its (-1)-eigenspace have the same dimension. Note that the statement "$A$ is a balanced observable" is equivalent to the statement "there exists an observable $B$ that anti-commutes with $A$". 

For two vectors $\ket{\varphi}, \ket{\psi}\in \mH$ and $\delta>0$ we write $\ket{\varphi}\approx_\delta \ket{\psi}$ to mean $\|\ket{\varphi}-\ket{\psi}\|=O(\delta)$, where the implicit constant is universal. If $\ket{\varphi_k}$, $\ket{\psi_k}$ are families of states indexed by a common integer $k$ we write $\ket{\varphi_k}\approx_\delta \ket{\psi_k}$ to mean $\Es{k} \|\ket{\varphi_k}-\ket{\psi_k}\|=O(\delta)$, where $\Es{k}$ denotes a uniformly random index $k$ in the allowed range.

\subsection{XOR games}
\label{sec:xor}

For integers $n,m$ an $n\times m$ XOR game $G$ is specified by a real $n\times m$ matrix, that we often also call $G$, such that $\sum_{i,j} |G_{i,j}|=1$. A strategy for the players in $G$ is given by finite-dimensional Hilbert spaces $\mH_\reg{A}$ and $\mH_\reg{B}$, a collection of $n$ observables $A_i\in \Lin(\mH_{\reg{A}})$ for the first player, $m$ observables $B_j\in \Lin(\mH_\reg{B})$ for the second player, and a state $\ket{\psi}\in\mH_\reg{A}\otimes \mH_\reg{B}$ (in any finite dimension). The bias of the strategy is defined as
$$\beta^*(G;A_i,B_j,\ket{\psi}):= \sum_{i,j} G_{i,j} \bra{\psi} A_i\otimes B_j \ket{\psi}.$$
The bias of a game is the maximum bias achievable of any finite-dimensional strategy: 
$$\beta^*(G):= \sup_{d,A_i,B_j,\ket{\psi}} \Big|\sum_{i,j} G_{i,j} \bra{\psi} A_i\otimes B_j \ket{\psi}\Big|,$$
where the supremum is taken over all integers $d$, observables $A_i,B_j$ in $\C^d$ and states $\ket{\psi}\in\C^d\otimes \C^d$. Given $\eps>0$ we say that a strategy $(A_i,B_j,\ket{\psi})$ in $G$ is $\eps$-optimal if $\beta^*(G;A_i,B_j,\ket{\psi}) \geq (1-\eps)\beta^*(G)$.

Tsirelson~\cite{Tsirelson:85b} proved the following fact that will be relevant for our analysis: for any collection $x_i,y_j \in \R^d$ of real unit vectors there exists observables $A_i,B_j \in \C^D$ for $D \leq 2^{\lfloor d/2\rfloor}$ and $\ket{\psi} = D^{-1/2} \sum_{i=1}^D \ket{i} \ket{i}$ such that $\bra{\psi} A_i \otimes B_j \ket{\psi} = x_i \cdot y_j$ for every $i,j$. (Tsirelson's construction is based on the use of a representation of the Clifford algebra.) This observation allows to prove that the following semidefinite relaxation of the bias is tight:
\begin{align}
\beta^*(G) = \textsc{SDP}(G) &= \,\,\,\sup \,\sum_{i,j} G_{i,j}\, x_i \cdot y_j \label{eq:xor-sdp}\\
&\qquad x_i,y_j \in \R^{m+n}\notag\\
&\qquad \|x_i\|=\|y_j\|=1.\notag
\end{align}
We refer to~\cite{CHTW04} for a proof of this fact.

\subsection{The $\CHSH(n)$ games}

Our results are based on the analysis of a family of games $\CHSH(n)$, parametrized by an integer $n\geq 2$. The game $\CHSH(n)$ has $n$ possible questions for Alice, and $n(n-1)$ for Bob. The game can be described as follows:
\begin{enumerate}
\item The referee selects an ordered pair $i< j \in \{1,\ldots,n\}$ uniformly at random.
\item The referee sends either $i$ or $j$ to Alice (with probability $1/2$ each), and either $(i,j)$ or $(j,i)$ to Bob (again with probability $1/2$ each).
\item The players provide answers $a,b\in\{0,1\}$ respectively.
\item The referee accepts the answers if and only if $a\oplus b = 1$ if Alice received question $j$ and Bob $(j,i)$, and $a\oplus b = 0$ otherwise. 
\end{enumerate}

More concretely, the game matrix for the $\CHSH(n)$ has $n$ rows indexed by integers $k\in\{1,\ldots,n\}$, $n(n-1)$ columns indexed by pairs $(i,j)\in\{1,\ldots,n\}^2$ such that $i\neq j$, and such that the entry in the $k$-th row and $(i,j)$-th column is $0$ if $k\notin \{i,j\}$, $\frac{-1}{2n(n-1)}$ if $i>j$ and $k=i$, and $\frac{1}{2n(n-1)}$ otherwise. 

If $n=2$ then $\CHSH(2)$ is the $\CHSH$ game, which corresponds to the $\CHSH$ inequality of Clauser et al.~\cite{Clauser:69a}. In general $\CHSH(n)$ can be understood as playing one of ${n \choose 2}$ possible $\CHSH$ games, parametrized by ordered pairs $(i<j)$. While Bob, who is given the pair $(i,j)$, ``knows'' which game is being played, Alice only has partial information. 

The family of games $\CHSH(n)$ was first introduced in~\cite{Slofstra11xor}, who showed that $\omega^*(\CHSH(n)) = \cos^2(\pi/8)$ and that optimal strategies require local dimension at least $2^{\lfloor n/2\rfloor}$. We recall an optimal strategy for the players in this game. 

\begin{lemma}[\cite{Slofstra11xor}, Proposition~7]
Let $(A_i)_{i\in\{1,\ldots, n\}}$ be a collection of $n$ anti-commuting observables on $\C^d$ for Alice, and $(B_{ij}:=((-1)^{j<i} A_i^T+ A_j^T)/\sqrt{2})_{i \neq j\in\{1,\ldots,n\}}$ be observables for Bob. Let $\ket{\psi}$ be the maximally entangled state in $\C^d\otimes \C^d$. Then the strategy given by $(A_i,B_{ij},\ket{\psi})$ has success probability $\cos^2(\pi/8)$ in $\CHSH(n)$. Furthermore, for any integer $n$ there exists such a strategy with $d=2^{\lfloor n/2 \rfloor}$.
\end{lemma}

For the game $\CHSH = \CHSH(2)$ very good results are known characterizing the structure of $\eps$-optimal strategies. We use the following lemma from~\cite{McKagueYS12rigidity} (see also~\cite[Lemma~4.2]{ReichardtUV13leash}). 


\begin{lemma}[CHSH rigidity]\label{lem:chsh-rigid}
Let $\delta>0$ and $(\{A_0,A_1\},\{B_0,B_1\},\ket{\psi})$ a $\delta$-optimal strategy in $\CHSH$. Then there exists local isometries $U,V$ and a state $\ket{\psi'}$ such that, letting $\ket{\phi^+} = (1/\sqrt{2})(\ket{00}+\ket{11})$ and $X,Z$ the single-qubit Pauli operators,
\begin{gather}
\| U\otimes V \ket{\psi} - \ket{\phi^+}\ket{\psi'}\| = O(\sqrt{\delta}),\label{eq:chsh-rigidity-epr}\\
\max\Big\{\|\big( A_0 - U^\dagger (X\otimes \Id) U  \big)\otimes \Id\ket{\psi}\|,\;\| \big(A_1 - U^\dagger (Z\otimes \Id) U  \big)\otimes \Id\ket{\psi}\|\Big\} = O(\sqrt{\delta}),\label{eq:chsh-rigidity-xz}
\end{gather}
and letting $\tilde{A}_0:=\frac{B_0+B_1}{|B_0+B_1|} $ and $\tilde{A}_1:=\frac{B_0-B_1}{|B_0-B_1|} $, 
\begin{gather}
\max\Big\{\big\|\big( A_0 \otimes \Id - \Id \otimes \tilde{A}_0 \big)\ket{\psi}\big\|,\; \big\|\big( A_1 \otimes \Id - \Id \otimes \tilde{A}_1\big)\ket{\psi}\big\|\Big\} =  O(\sqrt{\delta}),\label{eq:chsh-btoa}\\
\max\Big\{\big\|\Id\otimes\big( \tilde{A}_0 - V^\dagger (X \otimes \Id) V  \big)\ket{\psi}\big\|,\;\big\| \Id\otimes \big(\tilde{A}_1 - V^\dagger (Z \otimes \Id) V  \big)\ket{\psi}\big\|\Big\} = O(\sqrt{\delta}).\label{eq:chsh-tildea-xz}
\end{gather}
\end{lemma}

\subsection{Overlapping qubits}

The notion of ``overlapping qubits'' is introduced in~\cite{reichardt16overlapping}. Intuitively, a pair of overlapping qubits $i$ and $j$ is specified by two pairs of anti-commuting observables $\{X_i,Z_i\}$ and $\{X_j,Z_j\}$ such that $[P_i,Q_j] \approx 0$ for $P,Q\in\{X,Z\}$. The following theorem from~\ref{reichardt16testing} bounds the distance of partially overlapping qubits from exact qubits. 

\begin{theorem} \label{t:EPRstabilizersfrommath}
Let $\ket \psi$ be a state in $\mH_{\reg{A}} \otimes \mH_{\reg{B}}$.  Assume that for each $j \in \{1,\ldots, n\}$ there are observables $X_j, Z_j \in \Lin(\mH_{\reg{A}})$ and $X'_j, Z'_j \in \Lin(\mH_{\reg{B}})$ such that $\{X_j, Z_j\} = \{X'_j,Z'_j\}=0$, and for all $i \neq j$ and for all $P, Q\in\{X,Z\}$, 
\begin{equation*}
\max\big\{ \big\| [P_i, Q_j]\otimes \Id \ket \psi \|,\, \big\| \Id \otimes [P'_i,Q'_j] \ket{\psi}\big\|\big\}\leq \eta,
\end{equation*}
and
\begin{equation*}
\big\| P_j \otimes P_j' \ket \psi - \ket \psi \| \leq \eta \enspace .
\end{equation*}
Let
\begin{equation*}
\ket{\psi'} = \ket{\psi}_{\reg{AB}} \otimes \ket{\phi^+}_{\reg{A'A''}}^{\otimes n} \otimes \ket{\phi^+}^{\otimes n}_{\reg{B'B''}} \,\in\, \mH_{\reg{A}}\otimes (\C^2)^{\otimes n}_{\reg{A'}} \otimes (\C^2)^{\otimes n}_{\reg{A''}} \otimes \mH_{\reg{B}} \otimes (\C^2)^{\otimes n}_{\reg{B'}} \otimes (\C^2)^{\otimes n}_{\reg{B''}}
 \enspace, 
\end{equation*}
where $\ket{\phi^+} = \frac{1}{\sqrt{2}} (\ket{00}+\ket{11})$.
Then there exist observables $\hat{X}_i, \hat{Z}_i \in \Lin(\mH_{\reg{A}}\otimes (\C^2)^{\otimes n}_{\reg{A'}} \otimes (\C^2)^{\otimes n}_{\reg{A''}})$ and $\hat{X}'_i, \hat{Z}'_i \in \Lin(\mH_{\reg{B}} \otimes (\C^2)^{\otimes n}_{\reg{B'}} \otimes (\C^2)^{\otimes n}_{\reg{B''}})$ such that $\{\hat{X}_j, \hat{Z}_j\} = \{\hat{X}'_j,\hat{Z}'_j\}=0$, and for $i \neq j$ for $P,Q\in\{X,Z\}$, $[\hat{P}_i, \hat{Q}_j]= [\hat{P}_i', \hat{Q}_j'] = 0$ and 
$$\max\big\{\| \big( \hat{P}_j - P_j \otimes \Id_{\reg{A'A''}} \big) \otimes \Id_{\reg{BB'B''}}  \ket{\psi'} \|,\, \| \big( \Id_{\reg{AA'A''}}\otimes (\hat{P}_j' - P_j' \otimes \Id_{\reg{B'B''}} \big) \ket{\psi'} \| \big\}= O(n \eta),$$
 and furthermore
\begin{equation*}
\big\|\hat{P}_j \otimes \hat{P}_j' \ket{\psi'} - \ket{\psi'} \| = O(n \eta)
 \enspace .
\end{equation*}
\end{theorem}

The theorem has the following immediate corollary.

\begin{corollary} \label{t:EPRpairsfromstabilizers}
Under the assumptions of Theorem~\ref{t:EPRstabilizersfrommath}, for $D\in \{A,B\}$ there are unitaries $U_{\reg{DD'D''}} : \mH_{\reg{D}} \otimes (\C^2)^{\otimes n}_{\reg{D'}}  \otimes (\C^2)^{\otimes n}_{\reg{D''}} \rightarrow \mH_{\reg{D}} \otimes (\C^2)^{\otimes n}_{\reg{D'}}  \otimes (\C^2)^{\otimes n}_{\reg{D''}}$ and a state $\ket{\ancilla} \in \mH_\reg{A} \otimes (\C^2)^{\otimes n}_{\reg{A''}} \otimes \mH_\reg{B} \otimes (\C^2)^{\otimes n}_{\reg{B''}}$ such that 
\begin{align*}
\big\| U_\reg{AA'A''} \otimes U_\reg{BB'B"} \ket{\psi'} - \ket{\phi^+}_{\reg{A'}\reg{B'}}^{\otimes n} \otimes \ket{\ancilla}_{\reg{AA''BB''}} \big\| &= { O(n^{3/2} \eta)} \\
\big\| (U_{\reg{DD'D''}} ((X_j)_{\reg{D}} \otimes \Id_{D'D''}) U^\dagger_{\reg{DD'D''}} - (\sigma^x_j)_{\reg{D'}} \otimes \Id_{\reg{DD''}}) \otimes \Id_{other side} \ket{\phi^+}_{\reg{A'}\reg{B'}}^{\otimes n} \otimes \ket{\ancilla}_{\reg{AA''BB''}} \big\| &= {O(n^{3/2} \eta)} \\
\big\|(U_{\reg{DD'D''}} ((Z_j)_{\reg{D}} \otimes \Id_{D'D''}) U^\dagger_{\reg{DD'D''}} - (\sigma^z_j)_{\reg{D'}} \otimes \Id_{\reg{DD''}}) \otimes \Id_{other side} \ket{\phi^+}_{\reg{A'}\reg{B'}}^{\otimes n} \otimes \ket{\ancilla}_{\reg{AA''BB''}} \big\| &= {O(n^{3/2} \eta)} 
 \enspace .
\end{align*}
\end{corollary}

Theorem~\ref{t:EPRstabilizersfrommath} requires observables that exactly anti-commute. The following lemma 
shows that approximately anti-commuting observables are never far from exactly anti-commuting ones. 

\begin{lemma}\label{lem:exact-ac}
Let $X,Z$ be balanced obervables on a space $\mH_\reg{A}$ of even dimension and let $\ket{\psi} \in \mH_\reg{A}\otimes \mH_\reg{B}$ be such that $\|\{X,Z\}\otimes \Id \ket{\psi}\|\leq \eps$. Then there exists a balanced observable $\tilde{Z}$ on $\mH_\reg{A}$ such that 
$$\|(Z-\tilde{Z})\otimes \Id \ket{\psi}\| \leq \sqrt{3/2} \,\eps,$$
 and  
$$\{X, \tilde{Z}\} = 0.$$
\end{lemma}

\begin{proof}
Make a change of basis so that  
$$X = \begin{pmatrix} \Id & 0 \\ 0 & -\Id\end{pmatrix}\qquad\text{and}\qquad Z = \begin{pmatrix} A & C \\ C^\dagger & B \end{pmatrix},$$
where the size of the blocks is $dim(\mH_\reg{A})/2$. Making a further change of basis we may also assume that $C$ is diagonal with non-negative real entries; this change of basis comes from the singular value decomposition of $C$ and does not affect the form of $X$. With this notation, we get 
\begin{equation}\label{eq:ac-0}
\{X,Z\}^2 = \begin{pmatrix} 4A^2 & 0 \\ 0 & 4B^2 \end{pmatrix}.
\end{equation} 

Let 
$$\tilde{Z} = \begin{pmatrix} 0 & \Id\\ \Id & 0 \end{pmatrix}.$$
Then $\{X,\tilde{Z}\}=0$, and it remains to show that $\|(Z-\tilde{Z})\otimes \Id \ket{\psi}\| \leq \sqrt{3/2} \,\eps$. 

Using $Z^2 = \Id$, we get $C^2 = \Id - A^2$ and $C^2 = \Id-B^2$. Using $C^2 \leq \Id$ and our assumption that $C$ is diagonal with non-negative real entries, we get
\begin{equation}
(\Id-C)^2 \leq 2(\Id - C^2) 
\end{equation}
and from here we get 
\begin{equation}\label{eq:ac-1}
(\Id - C)^2 \leq 2 A^2, \qquad (\Id - C)^2 \leq 2 B^2.
\end{equation}
We can then bound
\begin{align*}
(Z-\tilde{Z})^2 &\leq 2\begin{pmatrix} A & 0 \\ 0 & B \end{pmatrix}^2 + 2 \begin{pmatrix} 0 & C-\Id \\ C^\dagger-\Id & 0\end{pmatrix}^2 \\
& \leq \frac{1}{2}\{X,Z\}^2 + \{X,Z\}^2, 
\end{align*}
where to bound the first term we used~\eqref{eq:ac-0}, and to bound the second we used~\eqref{eq:ac-1} and~\eqref{eq:ac-0}. 
The lemma follows by evaluating both sides of the operator inequality on $\rho_{\reg{A}} = \Tr_{\reg{B}} \ket{\psi}\bra{\psi}$. 
\end{proof}

\section{Upper bound: playing XOR games with low entanglement}
\label{sec:ub}

The goal of this section is to prove Theorem \ref{thm:ub}.

Let $(A_s,B_{t},\ket{\psi})$ be an optimal strategy for the players in $G$, where $\ket{\psi} \in \C^D \otimes \C^D$. Let 
$$x_s = \bra{\psi} A_s \otimes \Id ,\qquad y_t = \bra{\psi} \Id \otimes B_t ,$$
and observe that $x_s,y_t$ are complex $D^2$-dimensional unit vectors such that $\sum_{s,t} G_{s,t} x_s \cdot \overline{y_t} = \beta^*(G)$ (where for complex row vectors $x=(x_1,\ldots,x_l)$ and $y=(y_1,\ldots,y_l)$ we denote $x\cdot y = \sum_i x_i y_i$). Let $d$ be an integer and  $\{g_{k p}\}$, for $k\in\{1,\ldots,d\}$ and $p\in\{1,\ldots,D^2\}$, be independent and uniformly distributed in $\{1,-1,i,-i\}$. Define $x'_s,y'_t \in \C^d$ by 
$$ (x'_s)_k = \frac{1}{\sqrt{d}}\sum_p g_{kp} (x_s)_{p} ,\qquad (y'_t)_k = \frac{1}{\sqrt{d}}\sum_p g_{k p}(y_t)_p,$$
for $k=1,\ldots,d$. Let $\alpha$ be a real parameter distributed according to the hyperbolic secant distribution. The only property of this distribution relevant for the analysis is that it satisfies that for any $a>0$, $\Es{\alpha}[a^{i\alpha}]=2a-\Es{\alpha}[a^{2+i\alpha}]$. Using this relation we obtain
\begin{align}
\Es{}\Big[ \sum_{s,t} G_{s,t} \frac{x'_s}{\|x'_s\|}\|x'_s\|^{i\alpha} \cdot \frac{\overline{y'_t}}{ \|y'_t\|} \|y'_t\|^{i\alpha} \Big] &= 2\sum_{s,t} G_{s,t} x_s \cdot \overline{y_t} - \Es{}\Big[ \sum_{s,t} G_{s,t} x'_s\|x'_t\|^{1+i\alpha} \cdot \overline{y'_t} \|y'_t\|^{1+i\alpha} \Big].\label{eq:ub-1}
\end{align}
We bound the second term on the right-hand side. For this we interpret $x'_s \|x'_s\|^{1+i\alpha}$ and $y'_t \|y'_t\|^{1-i\alpha}$ as a vector solution to the semidefinite program~\eqref{eq:xor-sdp} associated to the XOR game $G$, where the vectors are infinite-dimensional complex vectors in $L_2(\C)$.\footnote{Although a priori $\textsc{SDP}(G)$ considers a supremum over real finite-dimensional vectors, the extension to infinite-dimensional complex vectors does not allow for a larger value, as the vectors can always be projected down to their finite-dimensional span, and made real by considering $x\mapsto \Re(X)\oplus\Im(x)$ and $y\mapsto \Re(y) \oplus (-\Im(y))$.} Let $z\in\C^{D^2}$ be any vector among the $x_s$, $y_t$, and $z'\in L_2(\C)$ the associated vector, $x'_s \|x'_s\|^{1+i\alpha}$ or $y'_t \|y'_t\|^{1-i\alpha}$. We compute
\begin{align*}
\|z'\|^2 &= \frac{1}{d^2} \Es{} \Big( \sum_{k=1}^d \big| \sum_p g_{k p} z_p\Big|^2\Big)^2\\
&=\frac{1}{d^2}\Es{}\Big[ \sum_{k,k'=1}^d\sum_{p,q,r,s=1}^D g_{kp} \overline{g_{kq}} g_{k'r} \overline{g_{k's}}\, z_p \overline{z_q} z_r \overline{z_s} \Big]\\
&=\frac{1}{d^2}\Big( d\sum_{p\neq q} |z_p|^2|z_q|^2 + d^2\sum_{p,q} |z_p|^2|z_q|^2\Big)\\
&\leq \frac{d^2+d}{d^2} \|z\|^4 = 1 + \frac{1}{d}. 
\end{align*}
Using that the objective value of~\eqref{eq:xor-sdp} scales linearly with the norm of the vectors, this bound lets us upper bound the modulus of the second term on the right-hand side in~\eqref{eq:ub-1} by $(1+1/d)\beta^*(G)$, so that
$$\Es{}\Big[ \sum_{s,t} G_{s,t} \frac{x'_s}{\|x'_s\|} \|x'_s\|^{i\alpha} \cdot \frac{\overline{y'_t}}{ \|y'_t\|} \|y'_t\|^{i\alpha} \Big] \geq \Big(1 - \frac{1}{d}\Big)\beta^*(G).$$
In particular there must exist a choice of $g_{kp}$ and $\alpha$ such that the complex $d$-dimensional unit vectors $\frac{x'_s}{\|x'_s\|} \|x'_s\|^{i\alpha}$ and $\frac{y'_t}{\|y'_t\|} \|y'_t\|^{-i\alpha}$ yield a value for the left-hand side of~\eqref{eq:ub-1} that is at least $(1-1/d)\beta^*(G)$. Decomposing these vectors into real and imaginary parts as described earlier we obtain a real vector solution of dimension $2d$ achieving bias $(1-1/d)\beta^*(G)$ in~\eqref{eq:xor-sdp}. Applying Tsirelson's construction (as described in Section~\ref{sec:xor}) yields observables in dimension $2^{d}$ achieving the same value in $G$, proving Theorem \ref{thm:ub}.  

\section{Lower bound: rigidity for the $\CHSH(n)$ games}
\label{sec:lb}

The goal of this section is to prove Theorem \ref{thm:main}.

We start by assuming without loss of generality that $\mH_\reg{A}$, $\mH_\reg{B}$ are finite dimensional Hilbert spaces of even dimension $d$ each, and that Alice's observables are balanced; this assumption is technically required in the proof, and can always be satisfied by taking the direct sum with a space of appropriate dimension on which the state $\ket{\psi}$ has no mass, and on which all operators are extended by taking the direct sum with an appropriate reflection.

\medskip

The proof of Theorem \ref{thm:main} has several steps, which we give in the following lemmas. Our first lemma shows that in any good strategy for the $\CHSH(n)$ game, Alice's observables must approximately pairwise  anti-commute. 

\begin{lemma}\label{lem:anti}
Let $n\geq 2$, $\eps>0$ and $(A_i,B_{ij},\ket{\psi})$ an $\eps$-optimal strategy in $\CHSH(n)$. For all $i < j$ let 
$$\tilde{A}_{ij}:=\frac{B_{ij}+B_{ji}}{|B_{ij}+B_{ji}|},\qquad \tilde{A}_{ji} := \frac{B_{ij}-B_{ji}}{|B_{ij}-B_{ji}|}.$$
Then the following hold: 
\begin{align}
 \Es{i< j} \|\{A_i,A_j\}\otimes \Id \ket{\psi}\| = O(\sqrt{\eps}),\qquad\Es{i < j} \|\Id\otimes \{\tilde{A}_{ij},\tilde{A}_{ji}\} \ket{\psi}\| = O(\sqrt{\eps}),
\label{eq:anti}\\
\max\Big\{\Es{i< j}\big\|\big( A_i \otimes \Id - \Id \otimes \tilde{A}_{ij}\big)\ket{\psi}\big\|,\;  \Es{i<j}\big\|\big(A_j \otimes \Id - \Id \otimes \tilde{A}_{ji}\big)\ket{\psi}\big\|\Big\} =  O(\sqrt{\eps}).\label{eq:chsh-atotildea}
\end{align}
\end{lemma}

\begin{proof}[Proof of Lemma~\ref{lem:anti}]
We first observe that the game $\CHSH(n)$ is equivalent to the following game:
\begin{enumerate}
\item The referee selects a pair $(i,j)\in\{1,\ldots,n\}^2$ such that $i< j$ uniformly at random;
\item The referee selects $(x,y)\in\{0,1\}^2$ uniformly at random;
\item If $x=0$ the referee sends $i$ to Alice, and if $x=1$ he sends her $j$. If $y=0$ he sends $(j,i)$ to Bob and if $y=1$ the referee sends him $(i,j)$. 
\item Upon receiving answers $(a,b)$ the referee accepts if and only if $a\oplus b = x\wedge y$. 
\end{enumerate}
For any $(i,j)\in\{1,\ldots,n\}^2$ such that $i < j$ let $\eps_{ij}$ be such that the players' strategy achieves a bias $(1-\eps_{ij})\beta^*(\CHSH)$ in the game, conditioned on the referee having selected the pair $(i,j)$ in the first step of the reformulation above. Then $\Es{i<j} [\eps_{ij}] = \eps$, and for any $i<j$  the strategy $(A_i,A_j,B_{ij},B_{ji},\ket{\psi})$ is an $\eps_{ij}$-optimal strategy in the $\CHSH$ game.

 The relation~\eqref{eq:chsh-atotildea} then follows directly from~\eqref{eq:chsh-btoa} from the CHSH rigidity lemma, Lemma~\ref{lem:chsh-rigid}, and concavity of the square root function. 
To prove~\eqref{eq:anti}, write
\begin{align*}
  \{ A_i, A_j \}\otimes \Id \ket{\psi} &\approx_{\sqrt{\eps_{ij}}}(A_i \otimes \tilde{A}_{ji} + A_j \otimes \tilde{A}_{ij}) \ket{\psi}\\
	&\approx_{\sqrt{\eps_{ij}}} (U\otimes V)^\dagger \big((X\otimes Z+Z\otimes X)\otimes \Id \big) (U\otimes V)\ket{\psi}\\
		&\approx_{\sqrt{\eps_{ij}}} (U\otimes V)^\dagger \big((X\otimes Z+Z\otimes X)\ket{\phi^+}\big)\otimes \ket{\psi'} \\
	&= 0,
	\end{align*}
	where the first line uses~\eqref{eq:chsh-btoa}, the second~\eqref{eq:chsh-rigidity-xz} and~\eqref{eq:chsh-tildea-xz} (here the isometries $U$ and $V$ are allowed to depend on the pair $(i,j)$), the third~\eqref{eq:chsh-rigidity-epr}  and the fourth is by definition of $\ket{\phi^+}$. Averaging over $i<j$ and using concavity of the square root function proves the first part of~\eqref{eq:anti}. The second part follows similarly (alternatively, from the first part using~\eqref{eq:chsh-atotildea}).
\end{proof}

Given $n$ pairwise perfectly anticommuting observables $A_1,\ldots,A_n$, we can define $m=\lfloor n/3 \rfloor$ pairs of observables 
$$X_k = iA_{3k-2} A_{3k-1}\qquad\text{and}\qquad Z_k = iA_{3k-1}A_{3k},$$
for $k=1,\ldots,\lfloor n/3 \rfloor$, such that $\{X_k,Z_k\}=0$ and $[P_k,Q_\ell]=0$ for $k\neq \ell$ and $P,Q\in\{X,Z\}$. The following lemma shows that essentially the same construction also works in the approximate case. 

\begin{lemma}\label{lem:ac-qubits}
Let $\delta>0$ and $A_1,\ldots,A_n$ and $A'_1,\ldots,A'_n$ be observables such that 
\begin{equation}\label{eq:ac-qubits-0}
\Es{i} \| (A_i\otimes \Id - \Id \otimes A_i')\ket{\psi} \| \leq \delta\qquad\text{and}\qquad \Es{i\neq j} \|\{A_i,A_j\} \otimes \Id \ket{\psi}\| \leq \delta.
\end{equation}
Then there exists $m = \lfloor n/3 \rfloor$ pairs of observables $X_k,Z_k$ and $X'_k,Z'_k$ such that for all $k\neq \ell$, $\{X_k,Z_k\}=\{X'_k,Z'_k\}=0$ and for $P,Q\in\{X,Z\}$, 
$$\Es{k\neq \ell}\| [P_k,Q_\ell]\otimes \Id \ket{\psi} \| = O({\delta}),\qquad \Es{k\neq \ell} \| \Id\otimes [P'_k,Q'_\ell] \ket{\psi} \| = O({\delta}),$$
and 
$$ \Es{k} \|(P_k \otimes \Id - \Id \otimes P'_k) \ket{\psi}\|= O({\delta}).$$
\end{lemma}

\begin{proof}
For $k\in \{1,\ldots,m\}$ we construct $X_k$, $Z_k$, $X'_k$, $Z'_k$ in two stages. First, apply Lemma \ref{lem:exact-ac} independently to $(A_{3k-1}, A_{3k-2})$ and to $(A_{3k-1}, A_{3k})$ to obtain $\tilde{A}_{3k-2}$ and $\tilde{A}_{3k}$ that exactly anti-commute with $A_{3k-1}$. Next, let $X_k=i \tilde{A}_{3k-2} A_{3k-1}$ and $\tilde{Z}_k = i A_{3k-1} \tilde{A}_{3k}$. Then $X_k$, $\tilde{Z}_k$ are balanced observables and they satisfy 
\begin{align*}
\{X_k, \tilde{Z}_k\} \otimes \Id \ket{\psi} &= - \{\tilde{A}_{3k-2}, \tilde{A}_{3k}\} \otimes \Id \ket{\psi} \\
&\approx - (\tilde{A}_{3k-2} A_{3k} + \tilde{A}_{3k} A_{3k-2}) \otimes \Id \ket{\psi} \\
&\approx - \tilde{A}_{3k-2} \otimes A'_{3k} \ket{\psi} - \tilde{A}_{3k} \otimes A'_{3k-2} \ket{\psi} \\
&\approx - A_{3k-2} \otimes A'_{3k} \ket{\psi} - A_{3k} \otimes A'_{3k-2} \ket{\psi} \\
&\approx -\{A_{3k-2} , A_{3k}\} \otimes \Id \ket{\psi} \approx 0,
\end{align*}
where the total error in the chain of approximations is at most 
\begin{multline*}
2 \sqrt{3/2} \|\{A_{3k-2}, A_{3k-1}\} \otimes \Id \ket{\psi}\| + 2 \sqrt{3/2} \|\{A_{3k-1}, A_{3k}\} \otimes \Id \ket{\psi}\| \\
+ 2 \| (A_{3k-2} \otimes \Id - \Id \otimes A'_{3k-2}) \ket{\psi} \| + 2 \| (A_{3k} \otimes \Id - \Id \otimes A'_{3k}) \ket{\psi} \| + \|\{A_{3k-2}, A_{3k}\} \otimes \Id \ket{\psi}\| .
\end{multline*}
In a similar manner we can define $X'_k=i \tilde{A}'_{3k-2} A'_{3k-1}$ and $\tilde{Z}'_k = i A'_{3k-1} \tilde{A}'_{3k}$. Then $X'_k$, $\tilde{Z}'_k$ are balanced observables and we can obtain a similar bound on $\|\Id \otimes \{X'_k, \tilde{Z}'_k\} \ket{\psi}\|$. 

In the second stage, we apply Lemma \ref{lem:exact-ac} to $X_k, \tilde{Z}_k$ to obtain exactly anti-commuting $X_k, Z_k$ such that $\|(Z_k - \tilde{Z}_k) \otimes \Id \ket{\psi}\| \leq \|\{X_k, \tilde{Z}_k\} \otimes \Id \ket{\psi}\|$. Similarly, we apply Lemma \ref{lem:exact-ac} to $X'_k, \tilde{Z}'_k$ and obtain exactly anti-commuting $X'_k, Z'_k$ such that $\|(Z'_k - \tilde{Z}'_k) \otimes \Id \ket{\psi}\| \leq \|\{X'_k, \tilde{Z}'_k\} \otimes \Id \ket{\psi}\|$. It remains to show that $X_k, Z_k ,X'_k Z'_k$ satisfy the conclusions of Lemma \ref{lem:ac-qubits}. 

For each $k \neq l$ and  $P,Q \in \{X,Z\}$ we can bound $\| [P_k,Q_\ell]\otimes \Id \ket{\psi} \|$, $\| \Id\otimes [P'_k,Q'_\ell] \ket{\psi} \|$, and $\|(P_k \otimes \Id - \Id \otimes P'_k) \ket{\psi}\|$ using a similar chain of approximations to the one above. What is important here is that there is a small constant $c$ such that for all $i \neq j$, the terms $\| (A_i\otimes \Id - \Id \otimes A_i')\ket{\psi} \|$ and $\|\{A_i,A_j\} \otimes \Id \ket{\psi}\|$ appear at most $c$ times in the different error bounds that we obtain from the chains of approximation. Therefore, we can average over $k\neq l$ and use the assumptions \eqref{eq:ac-qubits-0} to obtain the conclusions of Lemma \ref{lem:ac-qubits}. 
\end{proof}

We will also need a lemma that demonstrates that if a state is close to a tensor product of a number of EPR pairs and an ancilla, then the state has high entanglement entropy. 

\begin{lemma}\label{lem:entropy}
Let $\ket{\psi}_{\reg{AA'BB'}}$ be a state in $(\C^2 \otimes \C^2)^{\otimes r}_{AB} \otimes (\mH_{\reg{A'}}\otimes \mH_{\reg{B'}})$ such that 
\begin{equation}
\|\ket{\psi}_{\reg{AA'BB'}} - \ket{\phi^+}_{\reg{A}\reg{B}}^{\otimes r} \otimes \ket{\ancilla}_{\reg{A'}\reg{B'}}\| \leq \delta/2
\end{equation}
Then, $\ket{\psi}_{\reg{AA'BB'}}$ has entanglement entropy at least \[ r - 4 \delta r + 2 \delta \log(\delta) \]
\end{lemma} 

\begin{proof}
We will use $\rho$ with appropriate subscripts to denote reduced density matrices of $\ket{\psi}_{\reg{AA'BB'}}$ and $\sigma$ with appropriate subscripts to denote reduced density matrices of $\ket{\phi^+}_{\reg{A}\reg{B}}^{\otimes r} \otimes \ket{\ancilla}_{\reg{A'}\reg{B'}}$. 

We will show that 
\begin{equation}\label{eq:entropyeq1}
S(\rho_{\reg{A}}) \geq r - 2 \delta r + \delta \log (\delta)
\end{equation}
and 
\begin{equation}\label{eq:entropyeq2}
S(\rho_{\reg{AB}}) \leq 2 \delta r - \delta \log(\delta) ,
\end{equation}
which using strong subadditivity as
 \[ S(\rho_{\reg{AA'}}) \geq S(\rho_{\reg{AA'B}}) + S(\rho_{\reg{A}}) - S(\rho_{\reg{AB}}) \geq S(\rho_{\reg{A}}) - S(\rho_{\reg{AB}}) \]
will prove the result. 

The trace distance between $\ket{\psi} \bra{\psi}$ and $\ket{\phi^+}_{\reg{A}\reg{B}}^{\otimes r} \otimes \ket{\ancilla}_{\reg{A'}\reg{B'}} \bra{\phi^+}_{\reg{A}\reg{B}}^{\otimes r} \otimes \bra{\ancilla}_{\reg{A'}\reg{B'}}$ is at most $\delta$. Take partial trace and get that the trace distance between $\rho_{\reg{A}}$ and $\sigma_{\reg{A}}$ is at most $\delta$. Apply Fannes inequality to get the bound \eqref{eq:entropyeq1}. Similarly, the trace distance between $\rho_{\reg{AB}}$ and $\sigma_{\reg{AB}}$ is at most $\delta$. Apply Fannes inequality again and get the bound \eqref{eq:entropyeq2}. This completes the proof of Lemma \ref{lem:entropy}. 
\end{proof}

Theorem~\ref{thm:main} follows from Lemma~\ref{lem:anti}, Lemma~\ref{lem:ac-qubits}, Lemma~\ref{lem:entropy} and Theorem~\ref{t:EPRstabilizersfrommath}.

\begin{proof}[Proof of Theorem~\ref{thm:main}]
Let $(A_i,B_{ij},\ket{\psi})$ an $\eps$-optimal strategy in $\CHSH(n)$. Applying Lemma~\ref{lem:anti} followed by Lemma~\ref{lem:ac-qubits} gives operators satisfying the assumptions of Theorem~\ref{t:EPRstabilizersfrommath}, on expectation, with $\eta=O(\sqrt{\eps})$. Applying Markov's inequality followed by Turan's theorem, for any integer $r\leq m$ there exists a set $S\subseteq \{1,\ldots,m\}$ of size $|S|=r$ such that the associated operators $(X_k,Z_k)$ and $(X'_k,Z'_k)$ for $k\in S$ satisfy the required conditions pairwise up to an error $\eta = O(\sqrt{\eps} r)$. 
Applying Corollary~\ref{t:EPRpairsfromstabilizers},  we get that the first bound in the corollary holds with error $\delta = c r^{5/2} \eps^{1/2}$. We choose $r = \Theta(\eps^{-1/5})$ so that $\delta=1/100$ (say), apply Lemma \ref{lem:entropy} and get that the entanglement entropy of $\ket{\psi}$ is $\Omega(\eps^{-1/5})$. 
\end{proof}

\bibliography{../main}

\end{document}